\documentclass[11pt]{article}
\usepackage{amsthm}
\usepackage{graphicx}

\graphicspath{{./graphics/}}

\newtheorem{theorem}{Theorem}[section]
\newtheorem{lemma}[theorem]{Lemma}
\newtheorem{proposition}[theorem]{Proposition}

\title{Zero-Suppressed Computation:\\ A New Computation Inspired by ZDDs\footnote{a preliminary version}}

\author{Hiroki Morizumi\\
{\small Interdisciplinary Graduate School of Science and Engineering, Shimane University,}\\
{\small Shimane 690-8504, Japan}\\
{\small morizumi@cis.shimane-u.ac.jp}
}

\date{}

\begin{document}

\maketitle

\begin{abstract}
Zero-suppressed binary decision diagrams (ZDDs) are a data structure
representing Boolean functions, and one of the most successful variants
of binary decision diagrams (BDDs).
On the other hand, BDDs are also called branching programs in computational
complexity theory, and have been studied as a computation model. 
In this paper, we consider ZDDs from the viewpoint of computational
complexity theory.
Our main proposal of this paper is that we regard the basic idea of ZDDs 
as a new computation, which we call zero-suppressed computation.
We consider the zero-suppressed version of two classical computation models,
decision trees and branching programs, and show some results.
Although this paper is mainly written from the viewpoint of computational
complexity theory, the concept of zero-suppressed computation can be
widely applied to various areas.
\end{abstract}

\section{Introduction}

Zero-suppressed binary decision diagrams (ZDDs) are a data structure
representing Boolean functions, introduced by Minato~\cite{Mina93},
and one of the most successful variants of binary decision diagrams (BDDs).
Knuth has referred to ZDDs as an important variant of BDDs
in his book~\cite{Knu09}, and ZDDs are also referred to in
other books~\cite{MT98,Weg00}.
On the other hand, BDDs are also called branching programs in computational
complexity theory, and have been studied as a computation model. 
In this paper, we consider ZDDs from the viewpoint
of computational complexity theory.

ZDDs have the same shape as BDDs (and branching programs) have, and
the only difference is the way to determine the output.
An assignment to the variables determines a computation path from
the start node to a sink node.
ZDDs output 1 iff the value of the sink node is 1 and all variables
which are not contained in the computation path are assigned by 0.
(See Section~\ref{sec:pre} for the formal definitions.)
ZDDs have been considered to be effective for Boolean functions whose
outputs are almost 0.
Our main proposal of this paper is that we regard the basic idea of ZDDs 
as a new computation, which we call {\em zero-suppressed computation}.
We consider the zero-suppressed version of two classical computation models,
decision trees and branching programs.

The first computation model is decision trees.
Although decision trees appear in various areas, it is also a computation
model to compute Boolean functions in computational complexity theory.
We consider zero-suppressed computation for this model.
For randomized computation and quantum computation, variants of
decision trees (i.e., randomized decision trees and quantum decision
trees, respectively) have been well-studied.
We define zero-suppressed decision trees and show some gaps of
the complexity to deterministic decision trees.
Although our results for this model are quite simple observations,
it implies a difference between zero-suppressed computation and
other computations, and motivates the study of zero-suppressed computation.

The second computation model is branching programs.
Branching programs are known as a computation model to approach
the {\sf L}~vs.~{\sf P} problem.
It is known that the class of decision problems solvable by a nonuniform
family of polynomial-size branching programs is equal to
{\sf L/poly}~\cite{Cob66}.
{\sf L/poly} is the class of decision problems solvable by nonuniform
logarithmic space Turing machines.
If one have proven a superpolynomial lower bound for the size of
branching programs computing a Boolean function in {\sf P}, then
{\sf L} $\neq$ {\sf P}.
In this paper, we define zero-suppressed branching programs,
which actually have the same definition to (unordered) ZDDs,
and consider the following question:
Is the class of decision problems solvable by a nonuniform family of
polynomial-size zero-suppressed branching programs equal to {\sf L/poly}?
We prove three results which are related to the question.
Firstly, we prove that the class of decision problems solvable by
a nonuniform family of polynomial-size width 5 (or arbitrary constant
which is greater than 5) zero-suppressed branching programs is equal
to nonuniform {\sf NC$^1$}.
This corresponds to the well-known Barrington's theorem~\cite{Bar89},
which showed that the class of decision problems solvable by a nonuniform
family of polynomial-size width 5 branching programs is equal to nonuniform
{\sf NC$^1$}.
Secondly, we prove that the class of decision problems solvable by
a nonuniform family of polynomial-size zero-suppressed branching programs
contains {\sf L/poly}, and is contained in nonuniform {\sf NC$^2$}.
Thirdly, we prove that the class of decision problems solvable by
a nonuniform family of polynomial-size read-once zero-suppressed branching
programs is equal to the class of decision problems solvable by
a nonuniform family of polynomial-size read-once (deterministic) branching
programs.
When we prove the third result, we also give some insight of the reason
why the class of decision problems solvable by a nonuniform family of
polynomial-size zero-suppressed branching programs may not be equal
to {\sf L/poly} (Section~\ref{subsec:ro}).

\section{Preliminaries} \label{sec:pre}

A Boolean function is a function $f:\{0,1\}^n \to \{0,1\}$.

A {\em (deterministic) branching program} or {\em binary decision
diagram (BDD)} is a directed acyclic graph.
The nodes of out-degree 2 are called {\em inner nodes} and labeled
by a variable.
The nodes of out-degree 0 are called {\em sinks} and labeled by 0 or 1.
For each inner node, one of the outgoing edges is labeled by 0 and
the other one is labeled by 1.
There is a single specific node called the {\em start node}.
An assignment to the variables determines a computation path from
the start node to a sink node.
The value of the sink node is the output of the branching program or BDD.

A {\em zero-suppressed binary decision diagram (ZDD)} is also
a directed acyclic graph defined in the same way as BDD, and
the only difference is the way to determine the output.
An assignment to the variables determines a computation path from
the start node to a sink node.
The ZDD outputs 1 iff the value of the sink node is 1 and all variables
which are not contained in the computation path are assigned by 0.

Notice that we define BDD and ZDD with no restriction to the appearance of
the variables.
(In some papers, BDD and ZDD mean the ordered one, i.e., the variable
order is fixed and each variable appears at most once on each path.)
The {\em size} of branching programs is the number of its nodes.
If the nodes are arranged into a sequence of levels with edges going only
from one level to the next, then the {\em width} is the size of
the largest level.
A branching program is called {\em (syntactic) read-once} branching program
if each path contains at most one node labeled by each variable.

Decision trees can be defined along the definition of branching programs.
We use this way in this paper.
A {\em (deterministic) decision tree} is a branching program whose graph
is a rooted tree.
The start node of a decision tree is the root.
We define the {\em (deterministic) decision tree complexity} of $f$,
denoted by $D(f)$, as the depth of an optimal (i.e., minimal-depth)
decision tree that computes $f$.

For a nonnegative integer $i$, {\sf NC$^i$} is the class of decision problems
solvable by a uniform family of Boolean circuits with polynomial size,
depth $O(\log^i n)$, and fan-in 2.

\section{Zero-Suppressed Decision Trees}

In this section, we consider decision trees and zero-suppressed
computation.
We firstly define zero-suppressed decision trees and the zero-suppressed
decision tree complexity, and show gaps for the deterministic
decision tree complexity.

\subsection{Definitions}

Since a decision tree is a branching program whose graph is a rooted tree,
zero-suppressed decision trees are naturally defined as follows.

A {\em zero-suppressed decision tree} is also a rooted tree defined
in the same way as deterministic decision tree, and the only difference
is the way to determine the output.
An assignment to the variables determines a computation path from
the start node to a sink node.
The zero-suppressed decision tree outputs 1 iff the value of the sink node
is 1 and all variables which are not contained in the computation path
are assigned by 0.
We define the {\em zero-suppressed decision tree complexity} of $f$,
denoted by $Z(f)$, as the depth of an optimal (i.e., minimal-depth)
zero-suppressed decision tree that computes $f$.

\subsection{Gaps}

In the previous subsection, we have defined the zero-suppressed decision
tree complexity. We can immediately obtain the following gaps.

\begin{theorem}
There is a Boolean function $f$ such that $D(f) = 0$ and
$Z(f) = n$.
\end{theorem}

\begin{proof}
Let $f = 1$.
(See also Figure~\ref{fig:f}.)
\end{proof}

\begin{figure}[t]
  \begin{center}
    \includegraphics[scale=0.6]{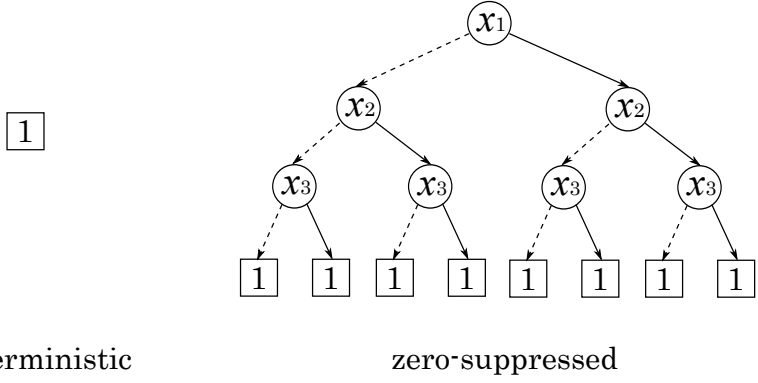}
    \caption{Decision trees computing $f$ for $n = 3$}
    \label{fig:f}
  \end{center}
\end{figure}

\begin{theorem}
There is a Boolean function $g$ such that $D(g) = n$ and
$Z(g) = 0$.
\end{theorem}

\begin{proof}
Let $ g = \neg x_{1} \land \neg x_{2} \land \cdots \land \neg x_{n}$.
(See also Figure~\ref{fig:g}.)
\end{proof}

\begin{figure}[t]
  \begin{center}
    \includegraphics[scale=0.6]{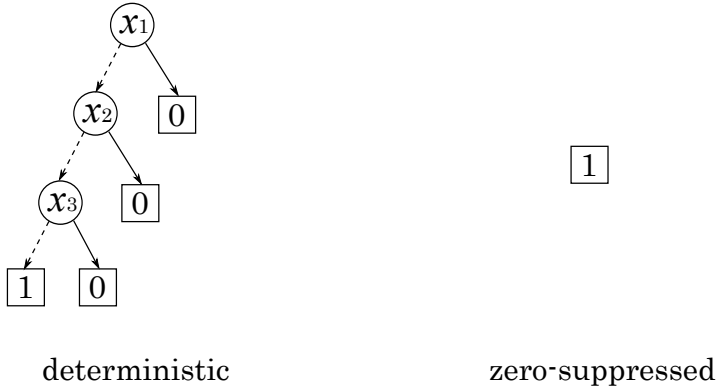}
    \caption{Decision trees computing $g$ for $n = 3$}
    \label{fig:g}
  \end{center}
\end{figure}

Thus, the advantages and disadvantages of deterministic and
zero-suppressed decision trees strongly depend on the Boolean function
which decision trees compute.
Although these two theorems are quite simple observations,
the difference from other computations implies unique behavior
of zero-suppressed computation.

\begin{proposition}
$Q_2(f) \leq R_2(f) \leq D(f)$.
\end{proposition}

$Q_2(f)$ and $R_2(f)$ are variants of decision tree complexity
with quantum computation and randomized computation, respectively.
For the definitions and the more details, we refer to Section~3
of the survey paper~\cite{BW02}.

\section{Zero-Suppressed Branching Programs}

In this section, we consider branching programs and zero-suppressed
computation.
Branching programs and BDDs have a same definition as we defined
in Section~\ref{sec:pre}.
Naturally, we define {\em zero-suppressed branching programs} as it has
the same definition to ZDDs.

\subsection{Constant-width zero-suppressed branching programs}

Firstly, we prove two lemmas, which are used also in the following subsection.

\begin{lemma} \label{lem:conv1}
Any deterministic branching program of $n$ variables, size $s$, and width $w$
can be converted to a zero-suppressed branching programs of size $s+n$ and
width $w$.
\end{lemma}

\begin{proof}
Let $G$ be a deterministic branching program of $n$ variables, size $s$,
and width $w$.
We convert $G$ to a zero-suppressed branching program as follows.
We add $n$ nodes, $v_1, v_2, \ldots, v_n$, such that $v_i$ is labeled
by $x_i$ for $1 \leq i \leq n$,
and connect two outgoing edges of $v_i$ to $v_{i+1}$ for $1 \leq i \leq n-1$,
and connect two outgoing edges of $v_n$ to the 1-sink,
and connect all edges which are connected to the 1-sink to $v_1$.

In the obtained zero-suppressed branching program, every computation path
to the 1-sink contains all variables.
Thus, by the definition of zero-suppressed branching programs, $G$ and
the obtained zero-suppressed branching program compute the same Boolean
function.
\end{proof}

\begin{lemma} \label{lem:conv2}
Any zero-suppressed branching programs of $n$ variables, polynomial size,
and width $w$ can be converted to a Boolean circuit of polynomial size and
depth $O(\log w \log n)$.
\end{lemma}

\begin{proof}
We extend the proof of one direction of the Barrington's theorem.
A deterministic branching program of $n$ variables, polynomial size,
and width 5 can be converted to a Boolean circuit of polynomial size and
depth $O(\log n)$ as follows.
Two levels of a deterministic branching program are composed to one level
by a circuit of a constant depth.
Doing this in parallel and repeating it $O(\log n)$ times
yield the desired circuit of depth $O(\log n)$.

If the width is $w$, two levels of a deterministic branching program are
composed to one level by a circuit of depth $O(\log w)$.
For the case of zero-suppressed branching programs, we need to memorize
the variables contained in the computation path, which can be done with
no increase of the depth of the circuit.

Actual encoding of each level is as follows.
At most $w$ nodes of each level can be numbered with $\lceil \log w \rceil$
bits.
For each outgoing edge of each node of a level,
$\lceil \log w \rceil + n$ bits are assigned.
The first $\lceil \log w \rceil$ bits represent the node which
the outgoing edge connects to.
The other $n$ bits represent whether each of $n$ variables is contained
in the computation path when the outgoing edge is used in computation.
\end{proof}

For the case that the width of zero-suppressed branching programs is
a constant, we determine that the equivalent class is {\sf NC$^1$},
which is an analog of the Barrington's theorem~\cite{Bar89} for
deterministic branching programs.

\begin{theorem}
For any constant $w \geq 5$, the class of decision problems solvable by
a nonuniform family of polynomial-size width $w$ zero-suppressed branching
programs is equal to nonuniform {\sf NC$^1$}.
\end{theorem}

\begin{proof}
All problems in nonuniform {\sf NC$^1$} can be solvable by a nonuniform family
of polynomial-size width 5 deterministic branching programs~\cite{Bar89}.
By Lemma~\ref{lem:conv1}, the problems can be solvable
also by a nonuniform family of polynomial-size width 5 zero-suppressed
branching programs.
Thus, the class contains nonuniform {\sf NC$^1$}.

Consider a problem solvable by a nonuniform family of polynomial-size
width $w$ zero-suppressed branching programs.
By Lemma~\ref{lem:conv2}, the problem is also solvable
by a nonuniform family of Boolean circuits of polynomial size and depth $O(\log w \log n)$.
Since $w$ is a constant, the class is contained in nonuniform {\sf NC$^1$}.
\end{proof}

\subsection{General zero-suppressed branching programs}

The main question for zero-suppressed branching programs is whether
the class of decision problems solvable by a nonuniform family of
polynomial-size zero-suppressed branching programs is equal to {\sf L/poly}
or not.
We show a weaker result.

\begin{theorem} \label{thrm:gene}
The class of decision problems solvable by a nonuniform family of
polynomial-size zero-suppressed branching programs contains {\sf L/poly},
and is contained in nonuniform {\sf NC$^2$}.
\end{theorem}

\begin{proof}
All problems in {\sf L/poly} can be solvable by a nonuniform family of
polynomial-size deterministic branching programs~\cite{Cob66}.
By Lemma~\ref{lem:conv1}, the problems can be solvable
also by a nonuniform family of polynomial-size zero-suppressed branching
programs.
Thus, the class contains {\sf L/poly}.

Consider a problem solvable by a nonuniform family of polynomial-size
zero-suppressed branching programs.
Obviously, the width of the zero-suppressed branching programs is
a polynomial of $n$.
Thus, by Lemma~\ref{lem:conv2}, the problem is also solvable
by a nonuniform family of Boolean circuits of polynomial size and depth $O(\log^2 n)$.
Therefore, the class is contained in nonuniform {\sf NC$^2$}.
\end{proof}

\subsection{Read-once zero-suppressed branching programs} \label{subsec:ro}

In deterministic branching programs, the states in computation are decided
only by the node which was reached in computation.
Thus, the number of the states is at most the size of the branching
program, and, if the size is at most polynomial, then each state can
be represented by logarithmic space, which leads to the fact that
the class of decision problems solvable by a nonuniform
family of polynomial-size deterministic branching programs is equal to
{\sf L/poly}.
On the other hand, in zero-suppressed branching programs, the states
in computation are not decided only by the node which was reached
in computation.
It depends on the variables which were contained in the computation path.
This is the main reason why the class of decision problems solvable
by a nonuniform family of polynomial-size zero-suppressed branching programs
may not be equal to {\sf L/poly}.
Note that the information of the passed variables cannot be saved in
logarithmic space.

In this subsection, we consider a simple case.
If deterministic and zero-suppressed branching programs are read-once,
then we can convert them to each other with polynomial increase of the size.

\begin{theorem}
The class of decision problems solvable by a nonuniform family of
polynomial-size read-once zero-suppressed branching programs is equal to
the class of decision problems solvable by a nonuniform family of
polynomial-size read-once deterministic branching programs.
\end{theorem}

\begin{proof}
We prove two lemmas.

\begin{lemma}
Any read-once deterministic branching program of $n$ variables and size $s$
can be converted to a read-once zero-suppressed branching program of
size $s+2ns$.
\end{lemma}

\begin{proof}
Note that the way of the proof of Lemma~\ref{lem:conv1} does not give
a read-once zero-suppressed branching program.
We need more consideration to the place where new nodes are added.

Let $G$ be a read-once deterministic branching program of $n$ variables
and size $s$.
Let $v_1, v_2, \ldots, v_s$ be the nodes in $G$ such that
$v_1, v_2, \ldots, v_s$ is a topologically sorted order.
We convert $G$ so that every computation path which reaches to a node
contains the same all variables, for each node from $v_1$ to $v_s$.
Assume that every computation path which reaches to $v_i$
contains the same variables for each $1 \leq i \leq k-1$.
We convert $G$ so that every computation path which reaches to $v_k$
contains the same variables as follows.
Let $X_i$ be the set of variables which are contained in computation
paths to $v_i$, for $1 \leq i \leq k-1$.
Let $X$ be the union of $X_j$ such that there is an edge from $v_j$
to $v_k$.
Let $X_i' = X - X_i$.
For every edge $e$ from $v_i$ to $v_k$, $1 \leq i \leq k-1$,
we add $|X_i'|$ nodes, $u_1, u_2, \ldots, u_{|X_i'|}$, such that the nodes
are labeled by the variables contained in $X_i'$,
and connect two outgoing edges of $u_j$ to $u_{j+1}$ for
$1 \leq j \leq |X_i'|-1$,
and connect two outgoing edges of $u_{|X_i'|}$ to $v_k$,
and connect $e$ to $u_1$.
Let $G'$ be the obtained branching program.
If computation paths to the 1-sink in $G'$ do not contain all variables,
we modify $G'$ to contain all variables by a similar way.

$G'$ is read-once, since added nodes are labeled by the variables
contained in $X_i'$.
In $G'$, every computation path to the 1-sink contains
all variables.
Thus, by the definition of zero-suppressed branching programs,
$G$ and $G'$ compute the same Boolean function.
The number of added node is at most $n$ for each edge.
\end{proof}

\begin{lemma}
Any read-once zero-suppressed branching program of $n$ variables and size $s$
can be converted to a read-once deterministic branching program of size $s+2ns$.
\end{lemma}

\begin{proof}
Let $G$ be a read-once zero-suppressed branching program of $n$ variables
and size $s$.
Let $v_1, v_2, \ldots, v_s$ be the nodes in $G$ such that
$v_1, v_2, \ldots, v_s$ is a topologically sorted order.
We convert $G$ so that every computation path which reaches to a node
contains the same all variables, for each node from $v_1$ to $v_s$.
Assume that every computation path which reaches to $v_i$
contains the same variables for each $1 \leq i \leq k-1$.
We convert $G$ so that every computation path which reaches to $v_k$
contains the same variables as follows.
Let $X_i$ be the set of variables which are contained in computation
paths to $v_i$, for $1 \leq i \leq k-1$.
Let $X$ be the union of $X_j$ such that there is an edge from $v_j$
to $v_k$.
Let $X_i' = X - X_i$.
For every edge $e$ from $v_i$ to $v_k$, $1 \leq i \leq k-1$,
we add $|X_i'|$ nodes, $u_1, u_2, \ldots, u_{|X_i'|}$, such that the nodes
are labeled by the variables contained in $X_i'$,
and connect the outgoing 0-edge of $u_j$ to $u_{j+1}$ for
$1 \leq j \leq |X_i'|-1$,
and connect the outgoing 0-edge of $u_{|X_i'|}$ to $v_k$,
and connect the outgoing 1-edge of $u_j$ to the 0-sink for
$1 \leq j \leq |X_i'|$,
and connect $e$ to $u_1$.
Let $G'$ be the obtained branching program.
If computation paths to the 1-sink in $G'$ do not contain all variables,
we modify $G'$ to contain all variables by a similar way.

$G'$ is read-once, since added nodes are labeled by the variables
contained in $X_i'$.
By the definition of zero-suppressed branching programs,
$G$ and $G'$ compute the same Boolean function.
The number of added node is at most $n$ for each edge.
\end{proof}

By the two lemmas, the theorem holds.
\end{proof}

\section{Conclusions and Open Problems}

In this paper, we proposed zero-suppressed computation as a new
computation.
For decision trees and branching programs, we could smoothly define
the zero-suppressed versions.
A challenging open problem is to seek another computation model whose
zero-suppressed version is meaningful.
When we consider other computation models (e.g., Boolean circuits),
it is a difficult and interesting problem even to define the appropriate
zero-suppressed version.

Although this paper is mainly written from the viewpoint of computational
complexity theory, the concept of zero-suppressed computation can be
widely applied to various areas.
We show an example.
The exactly-$k$-function $E^n_k(x_1, \ldots, x_n)$
is 1 iff $\Sigma_{i=1}^n x_i = k$.
In the standard formulas, an obvious representation of $E^3_1$ is
$$(x_1 \land \neg x_2 \land \neg x_3) \lor
  (\neg x_1 \land x_2 \land \neg x_3) \lor
  (\neg x_1 \land \neg x_2 \land x_3).$$
In formulas with zero-suppressed computation, $E^3_1$ is simply represented by
$$(x_1)^z \lor (x_2)^z \lor (x_3)^z,$$
where $( )^z$ is a new operation which we define from the concept of
zero-suppressed computation and $(f)^z$ is 1 iff $f=1$ and all variables
which are not contained in $f$ are assigned by 0.

For zero-suppressed branching programs, it remains open whether the class
of decision problems solvable by a nonuniform family of polynomial-size
zero-suppressed branching programs is equal to {\sf L/poly} or not.
We showed some related results to the question in this paper.
By Theorem~\ref{thrm:gene}, there are the following four cases.
\begin{itemize}
\item The class is equal to {\sf L/poly}.
\item The class is equal to nonuniform {\sf NC$^2$}.
\item The class is equal to another known complexity class between
      {\sf L/poly} and nonuniform {\sf NC$^2$}.
\item The class is not equal to any known complexity class.
\end{itemize}
Our observation for zero-suppressed decision trees implies unusual
properties of zero-suppressed computation, which makes us feel
the possibility of some new complexity class.

\bibliographystyle{plain}
\bibliography{zdd}

\end{document}